\documentclass[11pt,reqno]{amsart}

\pdfoutput = 1

\usepackage[english]{babel}
\usepackage{amsmath,amsfonts,amssymb,amsthm}
\usepackage{amsfonts}
\usepackage{amsxtra}

\usepackage{array,dcolumn}
\usepackage{graphicx}
\usepackage[hyperfootnotes=true,
        pdffitwindow=true,
        plainpages=false,
        pdfpagelabels=true,
        pdfpagemode=UseOutlines,
        pdfpagelayout=SinglePage,
        hyperindex,]{hyperref}

\usepackage[T1]{fontenc}
\usepackage[utf8]{inputenc}
\usepackage[activate={true,nocompatibility},spacing,kerning]{microtype}
\usepackage[sc,osf]{mathpazo}
\microtypecontext{spacing=nonfrench}

 \calclayout

\newcommand{\mathsym}[1]{{}}
\newcommand{\unicode}[1]{{}}

\numberwithin{equation}{section}

\theoremstyle{plain}
\newtheorem{theorem}{Theorem}

\newtheorem{corollary}[theorem]{Corollary}
\newtheorem{proposition}[theorem]{Proposition}
\numberwithin{theorem}{section}

\theoremstyle{definition}

\theoremstyle{remark}
\newtheorem{remark}[theorem]{Remark}

\begin{document}
\title[Stieltjes--Wigert matrix ensemble]{Global and local scaling limits for the $\beta = 2$ Stieltjes--Wigert random matrix ensemble}
\author{Peter J. Forrester}
\address{School of Mathematics and Statistics, 
ARC Centre of Excellence for Mathematical
 and Statistical Frontiers,
University of Melbourne, Victoria 3010, Australia}
\email{pjforr@unimelb.edu.au}

\begin{abstract}
The eigenvalue probability density function (PDF) for the Gaussian unitary ensemble has a well known analogy with the
Boltzmann factor for a classical log-gas with pair potential $- \log | x - y|$, confined by a one-body
harmonic potential. A generalisation is to replace the pair potential by $- \log |\sinh (\pi (x-y)/L) |$. The resulting PDF first
appeared in the statistical physics literature in relation to non-intersecting Brownian walkers, equally spaced  at time
$t=0$, and
subsequently in the study of quantum many body systems of the Calogero-Sutherland type, and also in Chern-Simons field theory.
It is an example of a determinantal point process with correlation kernel
 based on the Stieltjes--Wigert polynomials. We take up the problem of determining the moments of this ensemble, and find an exact
expression in terms of a particular little $q$-Jacobi polynomial. From their large $N$ form, the global density can
be computed. Previous work has evaluated the edge scaling limit of the correlation kernel in terms of the Ramanujan ($q$-Airy) function.
We show how in a particular $L \to \infty$ scaling limit, this reduces to the Airy kernel.
\end{abstract}
\maketitle

\section{Introduction}\label{s1}
Highly recognisable in mathematical physics is the probability density function (PDF)
\begin{equation}\label{1.1}
p_N^{(G)}(x_1,\dots,x_N) = {1 \over C_N^{(G)}} \prod_{l=1}^N e^{- x_l^2} \prod_{1 \le j < k \le N} (x_k - x_j)^2, \qquad x_l \in \mathbb R \: (l=1,\dots,N)
\end{equation}
with normalisation
\begin{equation}\label{1.2}
 C_N^{(G)} = 2^{-N(N-1)/2} \prod_{j=1}^N j !.
 \end{equation}
 In a somewhat disguised form (\ref{1.1}) first (to this author's knowledge) arose in the 1940 work of Husimi \cite{Hu40}, who was studying
 properties of the ground state wave function $\psi_0(x_1,\dots, x_N)$ for $N$ spin polarised fermions in one-dimension with
 a harmonic confining potential. For this problem, using dimensionless units, Husimi gave $\psi_0$ in the Slater determinant form
 \begin{equation}\label{1.3}
 \psi_0(x_1,\dots,x_N) = {1 \over \sqrt{N!}} \det \Big [ \phi_{j-1}(x_k) \Big ]_{j,k=1}^N,
  \end{equation}
  where, with $H_l(x)$ denoting the Hermite polynomials,
 \begin{equation}\label{1.4}  
  \phi_{l}(x) = {1 \over \sqrt{2^l l! \pi^{1/2}}} H_l(x) e^{- x^2/2}.
  \end{equation}
  To understand this formula, note that the Schr\"odinger equation for $N$ particles on a line with a confining harmonic potential factorises as the
  sum of $N$ single particle Schr\"odinger equations
  \begin{equation}\label{1.5}  
  - {1 \over 2} \Big ( {d^2 \over d x^2 } - x^2 \Big ) \phi_l(x) = \varepsilon_l   \phi_l(x), \qquad \varepsilon_l = l + 1/2 \: \: (l=0,1,\dots).
    \end{equation}
   For (\ref{1.5})  $\{  \phi_{l}(x) \}_{l=0}^{N-1}$ are the normalised wave functions corresponding to
   the lowest allowed energy levels in order; forming the determinant ensures that the many body wave function is anti-symmetric
  as required for fermions.
 
 Factoring the normalisation in (\ref{1.4}) from each row of (\ref{1.3}), and the Gaussian from each column shows 
  \begin{equation}\label{1.6}  
  \psi_0(x_1,\dots,x_N) =  \prod_{l=1}^N { e^{- x_l^2} \over  \sqrt{2^l l! \pi^{1/2}}}  \det \Big [ H_{j-1} (x_k) \Big ]_{j,k=1}^N.
  \end{equation}
  Introducing the monic Hermite polynomials $p_l^{(G)}(x)$ we have $H_l(x) = 2^l p_l^{(G)}(x)$and thus
    \begin{equation}\label{1.8}  
 \det \Big [ H_{j-1} (x_k) \Big ]_{j,k=1}^N = 2^{N (N - 1)/2} \det      \Big [ p_{j-1}^{(G)} (x_k) \Big ]_{j,k=1}^N.
\end{equation}
But for general monic polynomials $\{ p_l(x) \}_{l=0,1,\dots}$
\begin{equation}\label{1.9} 
 \det      \Big [ p_{j-1} (x_k) \Big ]_{j,k=1}^N = \det [ x_k^{j-1} ]_{j,k=1}^N = \prod_{1 \le j < k \le N} (x_k - x_j),
 \end{equation}
 where the first equality follows by successive elementary row operations to eliminate all but the leading monomial
 from each row, and the second equality is the Vandermonde determinant identity (see e.g.~\cite[Ex.~1.9 q.1]{Fo10}).
 Substituting (\ref{1.9}) in (\ref{1.8}), and the result in (\ref{1.6}), we see that
 \begin{equation}\label{1.10} 
 | \psi_0(x_1,\dots, x_N) |^2 = p_N^{(G)}(x_1,\dots, x_N),
 \end{equation}
 or in words, the square of the ground state wave function for spin polarised fermions in one-dimension with
 a harmonic confining potential is given by the PDF (\ref{1.1}).
 
 There are other interpretations of (\ref{1.1}) in theoretical/ mathematical physics. Let $X$ be an $N \times N$ standard
 complex Gaussian matrix, and define the random Hermitian matrix $H$ by $H = {1 \over 2} (X + X^\dagger)$.
 The set of such matrices is said to form the Gaussian unitary ensemble (GUE).
 It is a well known result that the eigenvalue PDF for the GUE is given by (\ref{1.1}); see e.g.~\cite[Prop.~1.3.4]{Fo10}.
 Exponentiating the product over pairs in (\ref{1.1}) shows
  \begin{equation}\label{1.11} 
  p_N^{(G)}(x_1,\dots,x_N) \propto e^{- \beta U(x_1,\dots,x_N)}, \qquad U:= {1 \over 2} \sum_{j=1}^N x_j^2 - \sum_{1 \le j < k \le N} \log | x_k - x_j|, \: \: \beta = 2,
   \end{equation}
   thus revealing an analogy with the Boltzmann factor for a classical gas, in equilibrium at inverse temperature $\beta = 2$, and interacting via a repulsive
   logarithmic pair potential, and an attractive one-body harmonic potential towards the origin. This analogy was used extensively in random
   matrix theory by Dyson \cite{Dy62}.
   
   As a further interpretation, consider $N$ Brownian walkers in one-dimension, and thus individually with a PDF
   $u_t(x^{(0};x)$ obeying the diffusion equation
   \begin{equation}\label{1.13}   
   {1 \over D} {\partial u \over \partial  t} = {1 \over 2}  {\partial^2 u \over \partial  x^2},
    \end{equation}
    starting from the points $\mathbf x^{(0)} = (x_1^{(0)},\dots, x_N^{(0)})$, $(x_1^{(0)} < \cdots < x_N^{(0)})$.
    The Karlin--MacGregor formula \cite{KM59} tells us that the PDF for the event the walkers arrive at
    $\mathbf x = (x_1,\dots, x_N)$ without intersecting is
   \begin{equation}\label{1.14}     
   G_t(\mathbf x^{(0)}; \mathbf x) = \det \Big [ u_t(x_j^{(0)};x_k) \Big ]_{j,k=1}^N.
 \end{equation}   
 Using this, it can be shown \cite{KT02} (see also \cite[Prop.~10.1.12]{Fo10}) that the PDF for the event that
 $N$ non-intersecting Brownian  walkers   all starting at the origin, arrive at position $\mathbf x$ after time $t$, 
 with the non-intersecting condition required for all times $T$ $(T \to \infty)$, is equal to (\ref{1.1}).
 
 Our interest in the present work is in the generalisation of (\ref{1.1}) (set $c=1$ and take $L \to \infty$)
    \begin{equation}\label{1.15}     
    p_N^{(SW_e)}(x_1,\dots,x_N)  = {1 \over C_{N,c}^{(SW_e)} (q)}  \prod_{l=1}^N e ^{- c x_l^2} \prod_{1 \le j < k \le N} \Big (
    \sinh ( \pi (x_k - x_j)/L) \Big )^2,
 \end{equation} 
 where
 \begin{equation}\label{1.15a}   
  C_{N,c}^{(SW_e)} (q) = N! 2^{-N ( N - 1)} \Big ( \sqrt{\pi \over c} \Big )^N q^{N^2 - 1/2} q^{- {1 \over 6} N ( 2N - 1)( 2 N + 1)}
  \prod_{j=1}^{N-1} (1 - q^j)^{N-j},
  \end{equation} 
  with
  \begin{equation}\label{1.16a}     
  q = e^{-1/(2k^2)}, \qquad k^2 = { c L^2 \over (2 \pi)^2}.
 \end{equation}   
 Changing variables 
  \begin{equation}\label{1.16b} 
 u_j = e^{{2 \pi \over L} ( x_j + {\pi \over L c} ) } = q^{-N} e^{2 \pi x_j / L}
  \end{equation}  
  we see that
   \begin{equation}\label{1.17}   
   p_N^{(SW_e)}(x_1,\dots,x_N) dx_1 \cdots dx_N =   p_N^{(SW)}(u_1,\dots,u_N) du_1 \cdots du_N ,
   \end{equation}
   where
 \begin{equation}\label{1.18}     
 p_N^{(SW)}(u_1,\dots,u_N) = {1 \over C_{N}^{(SW)}(q)}  \prod_{l=1}^N w^{(SW)}(u_l;q) 
 \prod_{1 \le j < k \le N} (u_k - u_j)^2, \quad u_l \in \mathbb R^+ \: (l=1,\dots,N).
 \end{equation} 
 In (\ref{1.18})
  \begin{equation}\label{1.19} 
w^{(SW)}(u;q)   = {k \over \sqrt{\pi}} e^{- k^2 ( \log u)^2} 
\end{equation}
and, with $q$ again given by  (\ref{1.16a}),
  \begin{equation}\label{1.19a}  
C_{N}^{(SW)}(q) = N!   q^{- {1 \over 6} N ( 2N - 1)( 2 N + 1)}
  \prod_{j=1}^{N-1} (1 - q^j)^{N-j}.
  \end{equation} 
  
  It is well known \cite{Sz75} that the polynomials on the half line $ u > 0$,
  orthonormal with respect to the weight function $w^{(SW)}(u;q)$, are the
  Stieltjes--Wigert polynomials
 \begin{equation}\label{1.20}  
  S_l(u;q) := {(-1)^l q^{l/2 + 1/4} \over
\{ (1 - q)(1-q^2) \cdots (1 - q^l) \}^{1/2}}
\sum_{\nu = 0}^l \bigg [ {l \atop \nu} \bigg ]_q
q^{\nu^2} (-q^{1/2} u)^\nu,
\end{equation}
where, with
 \begin{equation}\label{13a}
  [n]_q! := {(q;q)_n \over (1 - q)^n}, \qquad (u;q)_n := (1 - u) (1 - q u ) \cdots (1 - q^{n-1} u) ,
 \end{equation}
 the quantities
 \begin{equation}\label{13}
   \Big [ {n \atop m} \Big ]_q = { [n]_q! \over [n - m]_q! [m]_q!}, 
   \end{equation}
are the $q$-binomial coefficients.   
This is the reason for the label $(SW)$ in (\ref{1.18}); the label $(SW_e)$ used in 
(\ref{1.15}) is to indicate the underlying  Stieltjes--Wigert polynomials in exponential variables.

In relation to the PDF (\ref{1.1}) we have indicated four distinct interpretations in the context of
theoretical/ mathematical physics. All  have
analogues for the PDF (\ref{1.15}). These will be reviewed in Section \ref{S2}. It is also true that (\ref{1.15})
has a further  interpretation relative to  (\ref{1.1}) --- this is in relation to a partition function
which occurs in Chern--Simons field theory \cite{Ma04,Ti04}. The various applications have resulted in a number
of works studying properties of (\ref{1.15}); see for example
\cite{Fo89,Fo94x,Ma04a,dHT05,DT07, Ok07, BS07, Ti09, HY09, Ti10, ST10, BMS11,Ma11,TK14,Ti17,DIW19}.

Notwithstanding this previous literature, in light of known properties of (\ref{1.1}), there are still some fundamental properties 
of (\ref{1.15}) which remain to be investigated. Here we consider two of these. 
The first relates to the moments (in exponential variables) of the density for (\ref{1.15}), or equivalently the usual moments 
defined as power sum averages  of
 (\ref{1.17}), and their consequence in relation to the computation of the global density in the latter.
 We are motivated by the fact that the moments for (\ref{1.1}),
  \begin{equation}\label{1.23} 
  m_{2l}^{(G)} = \Big \langle \sum_{j=1}^N  x_j^{2l} \Big \rangle_{(G)}  \qquad (l \in \mathbb Z_{\ge 0}),
  \end{equation} 
  have the closed form hypergeometric evaluation \cite[Th.~8]{WF14} 
  \begin{equation}\label{1.24}  
  2^l m_{2l}^{(G)} = N { (2l)! \over 2^l l!} \, {}_2 F_1 ( - l, 1-N;2;2).
   \end{equation} 
   (With $x_j$ replaced by $|x_j|$ in (\ref{1.23}), this evaluation in fact extends to complex $l$ \cite{CMOS19}.)
 The second relates to the edge scaling limit of the correlation kernel  for (\ref{1.15}), and its relation to the well known Airy kernel
 specifying the edge scaling limit of the correlation kernel $K^{(G)}$ for (\ref{1.1}) \cite{Fo93c}
 \begin{align}\label{SA}
 K^{(G)}_{\rm edge}(X,Y) & := \lim_{N \to \infty} {1 \over \sqrt{2} N^{1/6}} K^{(G)}(x,y) \Big |_{x = \sqrt{2N} + X/\sqrt{2}N^{1/6} \atop
 y =  \sqrt{2N} + Y/\sqrt{2}N^{1/6} } \nonumber \\
 & = { {\rm Ai} (X)   {\rm Ai}' (Y)   -    {\rm Ai} (Y)   {\rm Ai}' (X)   \over X  - Y }.
 \end{align}
 The functional form of the edge scaling limit of the correlation kernel for (\ref{1.15}), $K^{(SW_e)}_{\rm edge}(X,Y) $,
 is known from
 \cite{Fo94x}, and is given in terms of the special function (well defined for $|q| < 1$)
 \begin{equation}\label{1.20y} 
A_q(z) := \sum_{\nu = 0}^\infty {q^{\nu^2} (- z)^\nu \over (q; q)_\nu} .
\end{equation}
The specific problem to be addressed is to identify scaling variables, with $X,Y$ dependent on $L$,
so that in the limit $L \to \infty$ the kernel $K^{(SW_e)}_{\rm edge}(X,Y) $ reduces to (\ref{SA}).

The moments of interest for the models ($SW_e$) and $(SW)$ are specified by
 \begin{align}\label{1.22}  
 m_l^{(SW_e)}  = \Big \langle \sum_{j=1}^N e^{{2 \pi \over L} (x_j + {\pi \over L c} ) l} \Big \rangle_{(SW_e)}
  =  \Big \langle \sum_{j=1}^N  u_j^l \Big \rangle_{(SW)} \qquad (l \in \mathbb Z).
 \end{align}
 We will show in Section \ref{S3} that (\ref{1.22}), like its GUE counterpart (\ref{1.23}),
   admits a hypergeometric evaluation.  This involves the
   $q$-generalisation of the Gauss ${}_2 F_1$ function
  \begin{equation}\label{14a+}
  {}_2 \phi_1 \Big ( {a_1, a_2 \atop b_1} \Big | q;z \Big ) = \sum_{n=0}^\infty {(a_1;q)_n (a_2; q)_n \over (q;q)_n (b_1;q)_n} z^n,
   \end{equation} 
  or equivalently the little $q$-Jacobi polynomial
   \begin{equation}\label{14c+}  
   p_n^{(lq\text{-}J)}(x;a,b|q) = {}_2 \phi_1  \Big ( {q^{-n}, a b q^{n+1} \atop a q} \Big | q; q x \Big ).
   \end{equation}    
  
  \begin{proposition}\label{P1}
  Let $q$ be given by (\ref{1.16a}) and let $l \in \mathbb Z^+$. We have
   \begin{align}\label{1.25}     
  {1 \over N} q^{N l}  m_l^{(SW_e)}   & = - {1 \over N} {( - q^{-1/2} )^l \over 1 - q^{-l}} \,
  {}_2 \phi_1  \Big ( {q^{l}, q^{-l} \atop q^{-1}} \Big | q^{-1}; q^{-N - 1} \Big )  \nonumber \\
  &  =  - {1 \over N} {( - q^{-1/2} )^l \over 1 - q^{-l}}  p_l^{(lq\text{-}J)}(q^{-N};1,q|q^{-1}).
  \end{align}
  \end{proposition}
  
  Simple manipulation of the series definition of $ {}_2 \phi_1$ on the RHS of (\ref{1.25}) shows that it is
  in fact a function of $(q^{1/2} - q^{-1/2})$, and of $q^{-N}$. This is in keeping with the well known
  fact \cite{Ma04} that the resolvent corresponding to (\ref{1.15}), and thus the moments, permit an
  expansion in $1/N^2$ upon setting $2k^2 = N/\lambda$ in (\ref{1.16a}) so that
  \begin{equation}\label{14c}
  q = e^{-\lambda/N}.
  \end{equation} 
  With this choice of $q$, the large $N$ scaling limit of (\ref{1.25}), and thus the leading
  term in the $1/N^2$ expansion, is almost immediate.
  
  \begin{corollary}\label{C1} 
  For $l \in \mathbb Z^+$ we have
    \begin{equation}\label{16}
 \mu_{l,0}^{(SW^*)} :=   \lim_{N \to \infty} { q^{N l} \over N}  m_l^{(SW_e)}    \Big |_{  q = e^{-\lambda/N} } =  {(-1)^l \over \lambda l} \, {}_2 F_1 (-l,l;1;e^{\lambda}).
  \end{equation} 
  \end{corollary}
  
  In Section \ref{S3.2} the result (\ref{16}) will be used to give a new derivation of the
  corresponding limiting scaled spectral density of the ensemble specified by (\ref{1.18}).

  The correlation kernel for (\ref{1.15}) is specified in Section \ref{S4.1}, and its bulk and edge scaling
  limits, already known from \cite{Fo94x}, are revised in Section \ref{S4.2}. In Section \ref{S4.3},
  a suitable asymptotic expansion known from \cite{HP15,Ha17} is used to deduce the sought $L \to \infty$ scaling
  limit of $K^{(SW_e)}_{\rm edge}$ reclaiming (\ref{SA}).
  
    \begin{proposition}\label{P1d}
  Let $K^{(SW_e)}_{\rm edge}(X,Y)$ be specified by (\ref{4.10}) below, and thus be determined by
  the function $A_q(z)$ as defined in (\ref{1.20y}) . Let
   \begin{equation}\label{16r}  
   \epsilon = {2 \pi^2 \over c L^2}, \quad X(x,L) = {L \over 2 \pi} \log {1 \over 4} - {L \over 2 \pi} \epsilon^{2/3} x, 
   \quad Y(y,L) = {L \over 2 \pi} \log {1 \over 4} - {L \over 2 \pi} \epsilon^{2/3} y.
   \end{equation}   
  We have
    \begin{equation}\label{16x}
 -  \lim_{L \to \infty}   {L \over 2 \pi}   \epsilon^{2/3} K^{(SW_e)}_{\rm edge}(X(x,L),Y(y,L)) =    K^{(G)}_{\rm edge}(x,y). 
   \end{equation}    
   \end{proposition}
  
  \section{Interpretations of the PDF (\ref{1.15})}\label{S2}
  \subsection{Boltzmann factor of a classical gas}\label{S2.1}
  It is immediate that the PDF (\ref{1.15}) can be written in a Boltzmann factor form analogous to (\ref{1.1})
  \begin{equation}\label{p1}
  p_N^{(SW_e)} \propto e^{-\beta U_L(x_1,\dots,x_N)}, \qquad U_L := {c \over 2} \sum_{j=1}^N x_j^2 - \sum_{1 \le j < k \le N}
  \log \Big | \sinh {\pi (x_k - x_j) \over L} \Big |, \: \: \beta = 2.
   \end{equation} 
   What remains is to interpret the pair potential in this expression; the one-body term is simply an harmonic attraction
   towards the origin as in (\ref{1.1}), with a scale factor $c$. For this consider the pair potential $\Phi(\mathbf r, \mathbf r')$
   due to 
   a two-dimensional unit charge in the plane at point $\mathbf r = (x,y)$, and another at point $\mathbf r' = (x',y')$. With
   $\mathbf r'$ regarded as fixed, this pair potential must satisfy the two-dimensional Poisson equation
   $\nabla^2  \Phi(\mathbf r, \mathbf r')  = - 2 \pi \delta(\mathbf r'  - \mathbf r)$. Require too that the charges are restricted
   (at first) to the strip $0 \le y < L$ in the plane, and subject to semi-periodic boundary conditions
   $ \Phi((x, y+L),(x',y')) =  \Phi((x, y),(x',y'))$. Equivalently $2 \pi y / L$ can be regarded as the angular position on a cylinder,
   and with $x$ corresponding to the height. The explicit form of $\Phi$ is (see e.g.~\cite[\S 2.7]{Fo10})
   \begin{equation}\label{p2}
    \Phi((x, y),(x',y')) =   - \log \bigg (  \Big | \sin ( \pi ( y - y') + i (x - x'))/L) \Big | (L/\pi) \bigg ).
    \end{equation} 
    Requiring that all charges, confined at first to the strip  $0 \le y < L$, be further confined to the $x$-axis (or
    $\theta = 0$ in the cylinder picture, varying only in their height) we have that $y = y'$ in (\ref{p2}),
    which is then recognised as the pair potential in (\ref{p1}). Note that in this circumstance (\ref{p2}) exhibits the
    large separation asymptotic behaviour
  \begin{equation}\label{p3}   
    \Phi((x, 0),(x',0)) \mathop{\sim}_{|x - x'| \to \infty} - {\pi | x - x'| \over L},
    \end{equation}
    which is in fact proportional to the Coulomb potential in one-dimension.
    
 \subsection{Ground state wave function}
  Notwithstanding the determinantal structure associated with $p_N^{(SW)}$, as evidenced
  by (\ref{1.18}) and (\ref{1.9}), there is  no free Fermi system with a ground state wave function
  equal to either (\ref{1.15}) or (\ref{1.17}). The essential point here is that the class of
  single-particle Schr\"odinger operators which permit wave functions of the form
  $\sqrt{w(x)} p_l(x)$ for some weight function $w(x)$ and orthogonal polynomials
  $\{ p_l(x) \}$ is extremely limited --- this is quantified by Bocher's theorem
  \cite{Bo29}.
  
  On the other hand, it turns out \cite{Fo94x} that not only $p_N^{(SW_e)}$, but also its
  $\beta$-generalisation in the sense of (\ref{p1}), has an interpretation of a squared ground
  state wave function for the particular many body Schr\"odinger operator of
  Calogero--Sutherland type (see \cite{KK09} for an extended account of this class of quantum
  many body systems)
  \begin{multline}\label{p4}
  \mathcal H = - \sum_{j=1}^N {\partial^2 \over \partial x_j^2} + a^2 \sum_{j=1}^N x_j^2 -
  {m a \over L} \sum_{1 \le i < j \le N} (x_i - x_j) \coth \Big ( {x_i - x_j \over 2 L} \Big ) \\
  + {m (m - 1) \over 2 L} \sum_{1 \le j < k \le N} {1 \over \sinh^2 ((x_j - x_k)/2 L)}, \qquad m :=\beta/2.
  \end{multline}  
  In particular when $\beta = 2$ as in (\ref{1.15}), the second pair potential term in (\ref{p4}) vanishes
  but the first such term remains. Denoting this term $\sum_{1 \le i < j \le N} V(x_i - x_j)$, we see that
  analogous to (\ref{p3})
   \begin{equation}\label{p5} 
  V(x_i - x_j)  \mathop{\sim}_{|x_i - x_j | \to \infty}  - {m a \over L}  |x_i - x_j|. 
  \end{equation}
  
  \subsection{Eigenvalue PDF on the space of complex positive definite matrices}
  From the viewpoint of the set of positive definite matrices $\mathcal P_N$ as an
  example of a Riemannian manifold, the squared geodesic distance between
  $A,B \in \mathcal P_N$, $(d(A,B))^2$ say, is given by (see e.g.~\cite{Ba07})
   \begin{equation}\label{AB1}  
  (d(A,B))^2 = \sum_{j=1}^N \Big ( \log \lambda_j (A^{-1} B) \Big )^2, 
   \end{equation}
   where $\{ \lambda_j(M) \}_{j=1}^N$ denotes the eigenvalues of $M$. If the eigenvalue of
   $A$ are parametrised as
     \begin{equation}\label{AB1a} 
   \{  e^{x_j}  \}_{j=1}^N,
  \end{equation}  
  and $B$ is the identity, (\ref{AB1}) simplifies to read
    \begin{equation}\label{AB2}   
  (d(A,I))^2 = \sum_{j=1}^N x_j^2.
    \end{equation}  
    
    In \cite{SHBV17}, a density for the invariant Riemannian volume element, $d \mu (A)$ say, on $\mathcal P_N$
   proportional to 
     \begin{equation}\label{AB2a} 
     e^{- c (d (A,I))^2}
   \end{equation}     
   is proposed. A decaying density function is in fact necessary:
   like the Lebesgue measure on $\mathbb R$, $d \mu (A)$ itself is not normalisable. For complex
   positive definite matrices the volume element  has the explicit form 
   \begin{equation}\label{AB2b}  
   d \mu (A) = {1 \over (\det A)^{N}} (dA),
    \end{equation} 
   where $(dA)$ denotes the Legesgue measure for the independent entries of $A$ (both real and imaginary
    part for the off diagonal entries). 
   The required invariance $d \mu (A) = d \mu (M^{1/2} A M^{1/2})$
   can be checked using
  \cite[Exercise 1.3 q.2]{Fo10}; contrast this to the case of $A \in$GL$(\mathbb C)$, as discussed in e.g.~\cite[\S 2.1]{FZ17},
  for which the exponent on the RHS of (\ref{AB2b}) is $2N$.

    Changing variables to the eigenvalues and eigenvectors, it is well known that the eigenvector
    contribution factorises (see e.g.~\cite[Prop.~1.3.4]{Fo10} and so can be integrated out. Finally, parametrising
    the eigenvalues according to (\ref{AB1a}) gives the PDF
     \begin{equation}\label{AB2c}    
     {1 \over Z_{N,c}} e^{- c \sum_{j=1}^N x_j^2} \prod_{1 \le j < k \le N} \Big ( \sinh ((x_k - x_j)/2) \Big )^2,
       \end{equation}  
  and is thus identical to (\ref{1.15}) with $L = 2 \pi$.  
  
  \begin{remark}
  1.~Of interest in \cite{SHBV17} is $\langle \sum_{j=1}^N x_j^2 \rangle_{(SW_e)} |_{L = 2 \pi} $. We see from the definitions
  that
      \begin{equation}\label{AB2d}  
 \Big  \langle \sum_{j=1}^N x_j^2    \Big \rangle_{(SW_e)}   \Big |_{L= 2 \pi} = - {d \over d c}  \log Z_{N,c}.
  \end{equation}    
  Since $   Z_{N,c} =  C_{N,c}^{(SW_e)} (q) |_{L = 2\pi}$, it follows from (\ref{1.15a}) that this can
  be computed exactly. \\
  2.~For recent further developments of the theme of this subsection, see \cite{ST20}.
  \end{remark}
    
   \subsection{Non-intersecting Brownian walkers with equal spacing initial condition}
   Consider the Brownian walker problem of the paragraph containing (\ref{1.13}). Inserting the
   explicit form of $u_t$ in (\ref{1.14}) shows
    \begin{equation}\label{q1}
    G_t(\mathbf x^{(0)}; \mathbf x) = \Big ( {1 \over 2 \pi D t} \Big )^{N/2} e^{- \sum_{j=1}^N ( x_j^2 + (x_j^{(0)})^2)/2 D t}
    \det \Big [ e^{ x_j^{(0)} x_k/ Dt} \Big ]_{j,k=1}^N.
     \end{equation}
     In the case of  the equal  spacing initial condition $x_j^{(0} = (j-1)a$ ($j=1,\dots,N$), use of the Vandermonde
     formula (\ref{1.8}) shows (\ref{q1}) simplifies to
    \begin{equation}\label{q1a}
     G_t(\mathbf x^{(0)}; \mathbf x) \Big |_{x_j^{(0)} = (j - 1)a} = \Big ( {1 \over 2 \pi D t} \Big )^{N/2}
     e^{- \sum_{j=1}^N ( x_j^2 + (N - 1) a x_j + (j-1)^2 a^2)/ 2 D t}
     \prod_{1 \le j < k \le N} 2 \sinh {a (x_k - x_j) \over 2 D t}.
     \end{equation}
     Notice that after the simple change of variables
     \begin{equation}\label{q1b}
     x_j \mapsto x_j - {(N - 1)a \over 2},
     \end{equation}
     and with
    \begin{equation}\label{q1c}     
    c = 1 / Dt,
  \end{equation}
  this is proportional to the Boltzmann factor (\ref{p1}) with $\beta = 1$. 
  
  To obtain (\ref{p1}) with $\beta = 2$, require that after arriving at positions $\mathbf x$
  in time $t$, the walkers return to the same equal spacing configuration of their
   initial condition in further time $t$. The corresponding PDF is
   $$
   {G_t(\mathbf x_0, \mathbf x)  G_t(\mathbf x, \mathbf x_0) \over
   G_{2t}(\mathbf x_0, \mathbf x_0) }  \bigg |_{x_j^{(0)} = (j - 1)a \: }.
   $$
 After the change of variables (\ref{q1b}), and with $c$ given by (\ref{q1c}), this is seen to
 reduce to (\ref{1.15}), reproducing too the explicit value of the normalisation (\ref{1.15a}) apart
 from a factor of $N!$ which is accounted for by the ordering $x_1 < \cdots < x_N$
 assumed in (\ref{q1}).
 
 \subsection{Chern--Simons partition function}
 Consider the Chern--Simons action with gauge group $U(N)$ of the 3-sphere,
 and with coupling strength $k/4 \pi$ (see e.g.~\cite{Ma04a}). It was shown in
\cite{Wi89} that the corresponding partition function, $Z_{\mathbb S^3}$ say,
of significance from the fact that it is a topological invariant, has the evaluation
 \begin{equation}\label{2.4.1}
 Z_{\mathbb S^3} = {1 \over (k + N)^{N/2}} \sum_{w \in P_N} \varepsilon (w) \exp \Big (
 - {2 \pi i \over k + N} \rho \cdot w(\rho) \Big ),
  \end{equation}
  where $P_N$ denotes the set of permutations of $\{1,2,\dots,N\}$, $ \varepsilon (w)$ is the
  signature of the permutation $w$ and $\rho = {1 \over 2} (N-1,N-3,\dots,-N+1)$ is the Weyl vector of SU$(N)$.
  The observation of \cite{Ma04} is that the sum in (\ref{2.4.1}) can be recognised as the
  determinant in (\ref{2.4}) below.  In this formula  the partition $\kappa$ is to have all parts equal to zero, implying
  that $Z_{\mathbb S^3}$ can be expressed in terms of the integral over the Boltzmann factor
  (\ref{p1}). One reads off too that the  coupling constants are related by
  $$
  {1 \over c} = {4 \pi i \over k + N}.
  $$

  \section{Proof of Proposition \ref{P1} and its consequences}\label{S3}
  \subsection{A Schur function average}
  Our approach to establishing (\ref{1.25}) requires knowledge of the evaluation of the average value of
  the Schur polynomial
   \begin{equation}\label{2.1}     
   s_\kappa(u_1,\dots,u_N) := {\det [ u_k^{\kappa_j + j - 1} ]_{j,k=1}^N \over \det [ u_k^{j-1} ]_{j,k=1}^N },
 \end{equation} 
 where $\kappa$ denotes the partition
    \begin{equation}\label{2.1a} 
    \kappa = (\kappa_1,\kappa_2,\dots, \kappa_N), \qquad \kappa_1 \ge \kappa_2 \ge \cdots \ge \kappa_N,
    \end{equation}
    each $\kappa_i$ being a non-negative integer,
 with respect to the PDF (\ref{1.18}). This can be found in the work of Dolivet and Tierz \cite{DT07}.
 We will provide a different derivation.
 
 \begin{proposition}\label{P2}
 Let $\kappa$ be a partition as in (\ref{2.1a}), and denote $| \kappa | = \sum_{j=1}^N \kappa_j$. We have
 \begin{equation}\label{9a}    
q^{N | \kappa | }  \langle     s_\kappa(u_1,\dots,u_N)  \rangle_{(SW)} = q^{- {1 \over 2} \sum_{l=1}^N \kappa_l^2}
 \prod_{1 \le j < k \le N} {1 - q^{- (\kappa_j - j - \kappa_k + k)} \over 1 -  q^{-( k - j )}}.
 \end{equation}
 \end{proposition}
 
 \begin{proof}
 The change of variables $u_l \mapsto q^{-N} u_l$ shows that
  \begin{equation}\label{9.1}
   \langle     s_\kappa(u_1,\dots,u_N)  \rangle_{(SW)} = q^{-N | \kappa | }    \langle     s_\kappa(u_1,\dots,u_N)  \rangle_{(SW^*)},
 \end{equation} 
 where $SW^*$ is specified  by the PDF
 \begin{equation}\label{14a} 
 {1 \over  C_{N}^{(SW^*)}(q)}  \prod_{l=1}^N u_l^{-N} w^{(SW)}(u_l;q) \prod_{1 \le j < k \le N} (u_j - u_k)^2
 \end{equation} 
 for suitable normalisation $  C_{N}^{(SW^*)}(q)$. The distinguishing feature of (\ref{14a}) is that it is unchanged by the mappings
 $u_l \mapsto 1/ u_l$. We will see later (Remark \ref{R2a}) that working with (\ref{14a}) gives rise to formulas which can be
 identified, up to replacing $q$ by $q^{-2}$, with those appearing in the computation of the average of a Schur
 polynomial with respect to a PDF generalising the eigenvalue PDF for Haar distributed random unitary matrices on the
 unit circle, 
 associated with the Rogers--Szeg\"o orthogonal polynomials.

 Evaluating the denominator in (\ref{2.1}), which is the Vandermonde determinant by the second equality in
 (\ref{1.9}), then using the same equality in the reverse direction, we see from (\ref{1.18}) and
 (\ref{2.1}) that
  \begin{multline}
  \langle     s_\kappa(u_1,\dots,u_N)  \rangle_{(SW^*)}  =  {1 \over C_{N}^{(SW^*)}(q)}  \\
 \times  \int_0^\infty d u_1 \cdots   \int_0^\infty d u_N \, \prod_{l=1}^N u_l^{-N} e^{- k^2 ( \log u_l)^2}
  \det [ u_k^{j-1} ]_{j,k=1}^N    \det [ u_k^{\kappa_{N-j+1} + j - 1}]_{j,k=1}^N .
  \end{multline}
  According to the Andr\'eief identity (see e.g.~\cite{Fo18}) this multiple integral simplifies to the determinant of
  single integrals
   \begin{equation}\label{2.2} 
   \langle     s_\kappa(u_1,\dots,u_N)  \rangle_{(SW^*)}  =  {N! \over C_{N}^{(SW^*)}(q)}
   \det \Big [ \int_0^\infty w^{(SW)}(u;q) u^{j-k+ \kappa_k - 1} \, du \Big ]_{j,k=1}^N.
 \end{equation}    
 
 Each integral in (\ref{2.2}) is a moment of the log-normal weight (\ref{1.19a}), with the well known
 evaluation
  \begin{equation}\label{2.3}  
  \int_0^\infty x^n w^{(SW)}(x;q) \, dx = q^{- (n+1)^2/2}, \qquad n \in \mathbb C.
   \end{equation}  
   Thus
  \begin{equation}\label{2.4}     
    \langle     s_\kappa(u_1,\dots,u_N)  \rangle_{(SW^*)} =   {N! \over C_{N}^{(SW^*)}(q)} 
    \det \Big [ q^{ - {1 \over 2} (j - k + \kappa_k)^2} \Big ]_{j,k=1}^N.
    \end{equation}
    Since
    $$
  q^{ - {1 \over 2} (j - k + \kappa_k)^2} = q^{- {1 \over 2} (j-1)^2 - {1 \over 2} ( \kappa_k - k + 1)^2} v_k^{j-1}, \qquad v_k = q^{- (\kappa_k - k + 1)}
  $$
  the determinant in (\ref{2.4}) can be reduced to the Vandermonde determinant. Evaluating the latter according to the second equality in
  (\ref{1.9}) gives
  $$
    \langle     s_\kappa(u_1,\dots,u_N)  \rangle_{(SW^*)} =   {N! \over C_{N}^{(SW^*)}(q)}  \prod_{l=1}^N q^{ - {1 \over 2} (l-1)^2 - {1 \over 2} (\kappa_l - l + 1)^2}
    \prod_{1 \le j < k \le N} ( q^{-(\kappa_k - k + 1)} -    q^{-(\kappa_j - j + 1)} ).
    $$
    Manipulating the product over $j < k$, and introducing the normalisation by the requirement that the RHS equals unity for $\kappa_l = 0$ ($l=1,\dots,N$)
    gives the stated result.
  \end{proof}
  
  \begin{remark}\label{R2a}
  In a particular notation for the Jacobi theta functions define
  \begin{equation}\label{t3}
  \theta_3(z;q) = \sum_{n=-\infty}^\infty q^{n^2} z^n.
  \end{equation}
  For a suitable normalisation $C_N^{(RS)}(q)$, where the superscript $(RS)$ stands for Rogers-Szeg\"o in keeping with the name of
  the underlying orthogonal polynomials, introduce the PDF on $\theta_l \in [0,2 \pi]$
  $(l=1,\dots,N)$
   \begin{equation}\label{12a}   
   {1 \over C_N^{(RS)}(q)} \prod_{l=1}^N   \theta_3(e^{i \theta_l} ;q)   \prod_{1 \le j < k \le N} | e^{i \theta_k} -  e^{i \theta_j}  |^2.
\end{equation}
It is shown in \cite{On81, AO84}   that
$$
\langle     s_\kappa( e^{i \theta_1}, \dots, e^{i \theta_N})  \rangle_{(RS)} = q^{  \sum_{l=1}^N \kappa_l^2}
 \prod_{1 \le j < k \le N} {1 - q^{2 (\kappa_j - j - \kappa_k + k)} \over 1 -  q^{2( k - j )}}.
 $$
 Comparison with (\ref{9a}) and (\ref{9.1}) shows
  \begin{equation}\label{12b}   
\langle     s_\kappa( e^{i \theta_1}, \dots, e^{i \theta_N})  \rangle_{(RS)} =    
   \langle     s_\kappa(u_1,\dots,u_N)  \rangle_{(SW^*)}  \Big |_{q \mapsto q^{-2}}.
   \end{equation}
   To understand (\ref{12b}), first note that with $z_j = e^{i \theta_j}$
   $$
 \prod_{1 \le j < k \le N} | e^{i \theta_k} -  e^{i \theta_j}  |^2 \, d \theta_1 \cdots     d \theta_N =
 i^{-N} \prod_{l=1}^N z_l^{-N} \prod_{1 \le j < k \le N} (z_k - z_j)^2 \, dz_1 \cdots d z_N,
 $$
 which formally is identical to  the portion of the PDF in (\ref{14a}) excluding $\prod_{l=1}^N w_l^{(SW)}(u_l;q)$.
 This is to be combined with the moment formula 
    \begin{equation}\label{3.11a}  
 \int_{|z| = 1} \theta_3(z;q) z^n \, {dz \over i z} = q^{(n+1)^2}, \qquad n \in \mathbb Z.
\end{equation} 
Since (\ref{3.11a}) is identical in value to (\ref{2.3}), but with $q \mapsto q^{-2}$ in the latter, (\ref{12b}) follows.

 \end{remark}
 
 The utility of Proposition \ref{P2} for the purposes of calculating the power sum average in 
 (\ref{1.22}) stems from the equality (see e.g.~\cite{Ma95})
   \begin{equation}\label{13a+} 
 \sum_{j=1}^N u_j^l = \sum_{r=0}^{{\rm min} \, (l-1,N-1)} (-1)^r s_{(l-r,1^r)} (u_1,\dots, u_N),
 \end{equation}
 where $(l-r,1^r)$ denotes the partition with largest part $\kappa_1 = l  - r$, and $r$ parts $(r \le N - 1)$ equal to
 1. Thus the average of each of the Schur polynomials on the RHS of (\ref{13a+}) can be read off from
 (\ref{9a}), and moreover identified in terms of particular $q$-binomial coefficients (\ref{13}).

   \begin{corollary}
   For $l - r \ge 1$, $r \le N - 1$, we have
     \begin{equation}\label{13b}  
     q^{Nl} \langle s_{(l-r,1^r)} (u_1,\dots,u_N) \rangle_{(SW)} =
     q^{- (l-r)^2/2 - r/2}   \Big [ {N + l - r - 1 \atop l} \Big ]_{q^{-1}}
   \Big [ {l- 1 \atop r} \Big ]_{q^{-1}} .
  \end{equation}     
\end{corollary}     

Summing (\ref{13b}) over $r$ according to (\ref{13a}) gives a $q$-series which can, after some
manipulation be recognised as a particular terminating ${}_2 \phi_1$ sum, so establishing
Proposition \ref{P1}.

\begin{remark}
1.~In keeping with (\ref{12b}), replacing $q$ by $q^{-2}$ on the RHS of (\ref{13b}) gives the value of
$ \langle s_{(l-r,1^r)} (e^{i \theta_1},\dots, e^{i \theta_N}) \rangle_{(RS)}$ \cite[Eq.~(4.5)]{AO84},
and this same replacement on the RHS of (\ref{1.25}) gives
$\langle \sum_{j=1}^N e^{ i \theta_j l} \rangle_{(RS)}$ \cite[Eq.~(4.6)]{AO84}. \\
2.~It has been remarked that the PDF specifying $(SW^*)$ is unchanged by the mappings
$u_l \mapsto 1/u_l$. Hence
  \begin{equation}\label{13d}
  q^{-N l}  m_{-l}^{(SW_e)}  =    q^{N l}  m_{l}^{(SW_e)},
  \end{equation}
  so extending (\ref{1.25}) to $l$ a negative integer. \\
3.~The little $q$-Jacobi polynomials in (\ref{14c+}), being orthogonal polynomials, satisfy
a 3-term recurrence \cite{KLS10}
\begin{equation}\label{13c}  
- q^{-N} p_l = A_l p_{l+1} - (A_l + A_{-l}) p_l + A_{-l} p_{l-1},
\end{equation}    
valid for $l \ge 1$, with initial conditions $p_0 = 1$, $p_1 = 1 - q^{-N}$, and
where
$$
A_l = {(q^{(l + 1)/2} - q^{- (l + 1)/2}) ( q^{l/2} - q^{-l/2}) \over
 ( q^{l+1/2} - q^{-l - 1/2})  ( q^l - q^{-l})}.
$$
It follows that the moments similarly satisfy a 3-term recurrence. For the GUE (\ref{1.1}), there
is a well known 3-term recurrence for the moments due to Harer and Zagier \cite{HZ86}.
For a derivation of the latter in the context of the Fermi gas interpretation of (\ref{1.1}),
and an extension, see \cite{Fo20}. \\
4.~Orthogonal polynomials from the Askey scheme have occurred in a number of
recent works in random matrix theory; see \cite{CMOS19,FL19a,FL19,ABGS20,GGR20}.
\end{remark}

\subsection{Scaled large $N$ limit}\label{S3.2}
It has already been noted in Corollary \ref{C1} that with $q$ given by (\ref{14c}) the moments as given
by (\ref{1.25}) admit a well defined large $N$ limit. In keeping with Remark \ref{R2a}, upon replacing  $\lambda \mapsto
-  \lambda/2$ this same expression gives the value of the limiting scaled moments for the Rogers-Szeg\"o PDF
(\ref{12a}) \cite{On81},
 \begin{equation}\label{16a}
  \mu_{l,0}^{(RS)} :=  \lim_{N \to \infty}  {1 \over N} \Big \langle \sum_{j=1}^N e^{i \theta_j l} \Big \rangle_{(RS)} =  \mu_{l,0}^{(SW^*)} \Big |_{\lambda \mapsto -  \lambda/2}.
  \end{equation}
 Defining the scaled densities
  \begin{equation}\label{17}
  {\rho}_{(1),0}^{(SW^*)}(x) = \lim_{N \to \infty} {1 \over N} \rho_{(1)}^{(SW^*)}(x) \Big |_{q = e^{-\lambda/N}}, \qquad
   {\rho}_{(1),0}^{(RS)}(\theta) = \lim_{N \to \infty} {1 \over N} \rho_{(1)}^{(RS)}(\theta) \Big |_{q = e^{-\lambda/N}}, 
   \end{equation}
   we have the relations to the scaled moments
  \begin{equation}\label{17a}  
   \mu_{l,0}^{(SW^*)}  = \int_0^\infty x^l    {\rho}_{(1),0}^{(SW^*)}(x) \, dx, \qquad
   \mu_{l,0}^{(RS)}  = \int_0^{2 \pi}  e^{i \theta l}   {\rho}_{(1),0}^{(RS)}(\theta) \, d \theta.
 \end{equation}   
 The second of these can be immediately inverted to give
   \begin{equation}\label{17b}  
 {\rho}_{(1),0}^{(RS)}(\theta) = 1 + {1 \over  \pi} \sum_{l=1}^\infty     \mu_{l,0}^{(RS)}   \cos \theta l.
  \end{equation}  
  Moreover, the explicit form of $\{ \mu_l^{(RS)} \}_{l=1}^\infty $ allows for the sum to be computed,
  with the result \cite{LSS81,On81}
    \begin{equation}\label{17c}  
 {\rho}_{(1),0}^{(RS)}(\theta) = {1 \over \pi \lambda} \log \bigg ( {1 - \cos \theta_c + 2 \cos \theta + 2 \cos( \theta/2)
 \sqrt{ 2 \cos \theta - 2   \cos \theta_c} \over 1 +  \cos \theta_c} \bigg ) \chi_{\cos \theta \ge  \cos \theta_c },
   \end{equation} 
   where $   \cos \theta_c = 2 e^{-\lambda/2} - 1$, and $\chi_A = 1$ for $A$ true, $\chi_A = 0$ otherwise.

There is no literal analogue of (\ref{17b}) for the inversion of the first relation in (\ref{17a}). Instead,
the standard strategy is to introduce the generating function
 \begin{equation}\label{17.1}
 G^{(SW^*)}(y)  = {1 \over y} + {1 \over y} \sum_{l=1}^\infty    \mu_{l,0}^{(SW^*)} y^{-l},
 \end{equation}
 and note from a geometric series expansion that
 \begin{equation}\label{17.1a} 
  G^{(SW^*)}(y)  =  \int_I {   {\rho}_{(1),0}^{(SW^*)}(x)  \over y - x} \, dx
 \end{equation}  
 (here $I$ denotes the interval of support of $ {\rho}_{(1),0}^{(SW^*)}$). The relation
 can be inverted by the Sokhotski-Plemelj formula to give
 \begin{equation}\label{17.1b}   
   {\rho}_{(1),0}^{(SW^*)}(x)  = \lim_{\epsilon \to 0} {1 \over 2 \pi i} \Big (
   G^{(SW^*)}(x - i \epsilon) -   G^{(SW^*)}(x + i \epsilon)  \Big ).
  \end{equation} 
 For this to be of practical use we require the closed form of $ G^{(SW^*)}$.
  
  \begin{proposition}
  We have
   \begin{equation}\label{17.2}
  G^{(SW^*)}(y) = - {1 \over \lambda y} \log \bigg ( {1 + y + \sqrt{(1+y)^2 - 4 y e^\lambda} \over 2 y e^{\lambda}} \bigg ).
    \end{equation} 
    \end{proposition}
    
    \begin{proof}
    It follows from (\ref{17.1}) and (\ref{16}) that
   \begin{equation}\label{18.1}   
   {d \over d y} \Big ( y G^{(SW^*)}(y) \Big ) = - {1 \over \lambda y} \sum_{l=1}^\infty (-1)^l \, {}_2 F_1(-l,l;1;e^\lambda) y^{-l}.
  \end{equation} 
  Recognising that
  $$
  {}_2 F_1(-l,l;1;e^\lambda) = P_l^{(0,-1)} (1 - 2 e^\lambda),
  $$
  where $ P_n^{(\alpha, \beta)} $ denotes the Jacobi polynomial in usual notation, the sum in (\ref{18.1}) can be computed
  according to the standard generating function for the latter. This gives
  \begin{align}\label{18.1a}   
   {d \over d y} \Big ( y G^{(SW^*)}(y) \Big ) & = - {1 \over 2 \lambda y}  \bigg ( - 1 + { y - 1 \over (1 + 2y  (1 - 2  e^\lambda) + y^2)^{1/2} }\bigg ) \nonumber \\
   & = - {1 \over \lambda} {d \over d y} \log \bigg (
   {1 + y + \sqrt{ (1 + y)^2 - 4 y e^\lambda \over 2 y e^\lambda} }\bigg ),
   \end{align}
  where the second equality can be verified by a direct calculation. The formula (\ref{17.2}) follows.
  \end{proof}
  
  We can apply (\ref{17.1b}) to now deduce the scaled density.
  
  \begin{corollary}
  With $z = 1 - 2 e^\lambda$ set $z_\pm = - z \pm ( z^2 - 1)^{1/2}$. We have
    \begin{equation}\label{20}
    {\rho}_{(1),0}^{(SW^*)}(x) = {1 \over \pi \lambda x} \arctan \bigg ( \sqrt{4 e^\lambda x - (1 + x)^2 \over 1 + x} \bigg )
    \chi_{z_- < x < z_+}.
  \end{equation}   
    \end{corollary}
    
    \begin{proof}
    After substituting (\ref{17.2}) in (\ref{17.1b}), the result follows from the identity
    $$
    \arctan z = {1 \over 2 i} \log { x - i \over x + i}.
    $$
    \end{proof}

    \begin{remark}
    1.~Upon the change of variables $x \mapsto e^{-\lambda} x$, the functional form (\ref{20}) has
    been derived using a loop equation formalism in \cite{Ma04}. Earlier, this same functional form
    was known from the computation of the density of the scaled zeros of the Stieltjes-Wigert polynomials
    \cite{CL99,KVA99}. \\
    2.~Let $\hat{S}_l(u;q)$ denote the Stieltjes-Wigert polynomials (\ref{1.20}) scaled so that the coefficient
    of $u^l$. The determinant structure implied by (\ref{1.18}) gives
    \cite[special case of Prop.~5.1.4]{Fo10}
    $$
    \Big \langle \prod_{l=1}^N ( x- x_l) \Big \rangle_{(SW)} = \hat{S}_N(x;q).
    $$
    With $q$ given by (\ref{14c}), for large $N$ the LHS can to leading order be expressed in terms of the
    scaled density $  {\rho}_{(1),0}^{(SW)}$ (see e.g.~the discussion in the paragraph below (3.3) of
    \cite{FL15}) to give
  \begin{equation}\label{20.1}    
  \hat{S}_N(x; e^{-\lambda/N}) \sim \exp \bigg ( N \int_I \log (x - y)   {\rho}_{(1),0}^{(SW)}(y) \, dy \bigg ).
  \end{equation} 
  Let us now combine (\ref{20.1}) with knowledge of the fact that $\hat{S}_N$ obeys the second order
  $q$-difference equation (see e.g.~\cite{CK05}) 
   \begin{equation}\label{21}  
   f(xq) - {1 \over x} f(x) + {1 \over x} f(x/q) = q^N f(x).
   \end{equation}  
  Dividing both sides of (\ref{21}) by $f(x)$, substituting (\ref{20.1}) and expanding
  $q = (1 - \lambda/N + O(1/N^2))$ shows
   \begin{equation}\label{22} 
   e^{-u} +  e^{u -y}= e^{-\lambda} + e^{-y}, \qquad e^y := x, \: \: e^u := \exp \bigg ( \lambda x \int_I {  {\rho}_{(1),0}^{(SW)}(t)
   \over x - t} \, dt \bigg ).
    \end{equation} 
    With a different derivation (and slightly varying notation), this functional equation has been derived previously \cite{HY09,Ma04}.  \\
    3.~As commented below Proposition \ref{P1}, the LHS of (\ref{1.25}) can be expanded in a series in $1/N$. Thus, extending the
   notation for the leading term in (\ref{16}), we have that for large $N$
    \begin{equation}\label{22a}
   {q^{N l} \over N} m_l^{(SW_e)} \Big |_{q = e^{-\lambda/N}} = \mu_{l,0}^{(SW^*)} + {1 \over N^2}  \mu_{l,2}^{(SW^*)} +  {1 \over N^4}  \mu_{l,4}^{(SW^*)} + \cdots.
     \end{equation} 
   For the Rogers-Szeg\"o moments as specified by the average in (\ref{16a}), the analogous expansion up to this order has
   been calculated in \cite{On81}. The theory of Remark \ref{R2a} tells us that replacing $\lambda$ by $-2 \lambda$ in the latter gives
   the terms in (\ref{22a}). In particular, we read off from \cite{On81} that
    \begin{equation}\label{22b}
    \mu_{l,2}^{(SW^*)} = (-1)^l {\lambda^2 \over 24} \sum_{p=1}^l {e^{\lambda p} - 1 \over \lambda p} \binom{l}{p} \binom{l+p-1}{l}
    \Big ( (2p-1) l^2 - 2 p^2 \Big ). 
\end{equation}     
     \end{remark}
    
\section{Special functions and the limiting correlation kernel}   
\subsection{Stieltjes--Wigert correlation kernel}\label{S4.1}
 The PDF (\ref{1.18}) specifying the ensemble $(SW)$ is of a standard form familiar in the study
 of  unitary invariant random matrices; see \cite[Ch.~5]{Fo10}. In particular, the $k$-point correlation
 function $\rho_{(k)}$, defined as $N!/(N-k)!$ times the integral over all but the first $k$ co-ordinates
 is therefore given as a determinant
  \begin{equation}\label{23}  
\rho_{(k)}^{(SW)}(u_1,\dots,u_k) = \det \Big [   K_N^{(SW)}(u_j,u_l) \Big ]_{j,l=1}^k,
 \end{equation} 
 where $K_N^{(SW)}(x,y)$ is referred to as the correlation kernel. The latter is given in terms of the orthonormal polynomials
 associated with the weight function $w^{(SW)}(u)$, i.e.~the Stieltjes--Wigert polynomials (\ref{1.20}), by the sum
  \begin{align}\label{24}  
 K_N^{(SW)}(u,v) & =   \Big ( w^{(SW)}(u) w^{(SW)}(v)   \Big )^{1/2} \sum_{j=0}^{N-1} S_j(u;q)  S_j(v;q) \nonumber \\
 & =     \Big ( w^{(SW)}(u) w^{(SW)}(v)   \Big )^{1/2}  {C_N \over C_{N-1}} 
 {S_N(u;q) S_{N-1}(v;q) - S_{N-1}(u;q) S_{N}(v;q)   \over   u  - v },
 \end{align}
 where the second line follows by the Christoffel--Darboux summation formula (see e.g.~\cite[Prop.~5.1.3]{Fo10}). In this
 formula $C_N$ is the coefficient of $u^N$ in $S_N(u;q)$, which we read off from (\ref{1.20}) to be given by
  \begin{equation}\label{24a} 
  C_N = {q^{N^2 + N + 1/4} \over ( (1 - q) \cdots (1 - q^N))^{1/2}}.
   \end{equation}
   The change of variables (\ref{1.16b}) shows that the correlation kernel for the ensemble $(SW_e)$ is related to (\ref{24})
   by 
 \begin{equation}\label{24b}   
  K_N^{(SW_e)}(x,y) = {2 \pi \over L} \Big ( u v \, w(u;q) w(v;q) \Big )^{1/2}  K_N^{(SW)}(u,v) \Big |_{u =  q^{-N} e^{2 \pi  x/ L } \atop
  v =  q^{-N} e^{2 \pi  y /L }  }.
 \end{equation}  
 
 \subsection{Scaling limits}\label{S4.2}
 With $L$ fixed and $N$ large, the interpretation of (\ref{1.15}) in terms of the Boltzmann factor (\ref{p1}) gives the prediction \cite{Fo94x}
 that to leading
 order the density will be supported on the interval $[-\pi  N/Lc, \pi N /L c]$ and is on average uniform. Hence there are two distinct scaling limits.
 One is when the particles are located in the interior of the support and away from the edges. To accomplish this we choose to locate
 the particles in the neighbourhood of the origin. The other is when the particles are located a finite distance from one of the edges. We consider
 each separately.
 
 \subsubsection{Bulk scaling}
 Define
 \begin{align} 
 \ell (y;q) & = e^{\xi \pi y / 2L} q^{-1/16} \sum_{\nu=-\infty}^\infty (-1)^\nu q^{(\nu + 1/4 - c L y / 2 \pi)^2}  \nonumber \\
 & = e^{(\xi + 1) \pi y / 2L}  e^{- c y^2/2} \theta_3 ( -q^{1/2} e^{\pi y / L}; q),  \label{L1a}  \\
\hat{\ell}(y;\hat{q}) & =  e^{-\xi \pi y / 2L}  \theta_1\Big ( \pi \Big ( {\pi \over 4} + {c L y \over 2 \pi} \Big ) \Big |     \hat{q}    \Big ), \label{L1b}
\end{align}
 where in (\ref{L1b}) we have adopted the particular notation for the Jacobi theta function
 \begin{equation}\label{L1c}
 \theta_1(x|q) := -i \sum_{n=-\infty}^\infty (-1)^n q^{(n+1/2)^2} e^{2  i (n + 1/2)x}
  \end{equation}
  (note that this convention differs from that used in the second line of (\ref{L1a}) as specified by (\ref{t3})).   
  The bulk scaling limit of the kernel (\ref{24b}) has been established in terms of these functions in
  \cite{Fo94x}.
  
  \begin{proposition}\label{p4.1}
  Let $\xi = 1$ for $N$ even, and $\xi = -1$ for $N$ odd. Also define
   \begin{equation}\label{24c}   
 \hat{q} =    e^{- c L^2/2} .
    \end{equation}
    We have
 \begin{align} \label{4.10}
& K^{(SW_e)}_{\rm bulk}(x,y) := \lim_{N \to \infty} K^{(SW_e)}(x,y) \nonumber \\
&  \qquad     = \Big ( {c \over \pi} \Big )^{1/2} {1 \over (q;q)_\infty^3}
{\ell(x;q) \ell(-y;q) - \ell (y;q) \ell(-x;q) \over 2 \sinh (\pi (x - y)/L) }  \nonumber  \\
& \qquad = {1 \over L} {1 \over \theta_1'(0|\hat{q}^2)}
{\hat{\ell}(x;\hat{q}) \hat{\ell}(-y;\hat{q}) - \hat{\ell}(y;\hat{q}) \hat{\ell}(-x;\hat{q}) \over 2 \sinh (\pi (x - y)/L) } .
\end{align}
 \end{proposition}   
 
 \begin{remark}
 1.~The first of the equalities in (\ref{4.10}), in the case $N$ odd, has also been obtained in \cite{TK14}.
 Moreover, asymptotic estimates from this latter reference show that rate of convergence to the limit
 has a correction term proportional to $q^N$, and is thus exponentially small. \\
 2.~The bulk density is obtained by taking the limit $y \to x$ in (\ref{4.10}). The resulting expression is a non-trivial periodic
 function of period $2 \pi/cL$. Indeed a crystalline phase is a characteristic of the many body state resulting from the $|x|$ potential
 in one-dimension \cite{BL75}. \\
 3.~It follows from the final equality in (\ref{4.10}) and (\ref{L1b}) and (\ref{L1c}) that
   \begin{equation}\label{25}    
   \lim_{L \to \infty} {2 \pi \over c L} K_{\rm bulk}^{(SW_e)}\Big ( {2 \pi \over c L} x, {2 \pi \over c L} y \Big ) = {\sin \pi (x - y) \over \pi (x - y)},
   \end{equation}
   (this result is also deduced in \cite{TK14} by a somewhat complicated calculation using working based on the first of the evaluations in (\ref{4.10})). The kernel
   in (\ref{25}) is well known as specifying the bulk scaling limit of Hermitian random matrices; see e.g.~\cite{EY17}.
 \end{remark}
 
 \subsubsection{Edge scaling}
 Writing $x = - \pi N / L c + X$ as appropriate for the analysis of the (left) edge scaling (recall the discussion at the beginning of
 \S \ref{S4.2}), we see from (\ref{24b}) and (\ref{24}) that relevant is the large $N$ form of $S_N(z;q)$ with $z,q$ fixed. In relation to 
 this, it is easy to check from the definition (\ref{1.20}) that
\begin{equation}\label{1.20x}  
S_N(z;q) = {(q;q)_N^{1/2} \over (q;q)_\infty} (-1)^N q^{N/2 + 1/4} \Big (
A_q(q^{1/2} z) - {q^{1 + N} \over 1 - q}  A_q(q^{-1/2} z)  + O(q^{2 N}) \Big ),
\end{equation}
where $A_q(y)$ is specified by (\ref{1.20y}).
In fact the limit implied by the leading term is already in Wigert's original paper \cite{Wi23}.

From the expansion (\ref{1.20x}) the scaling limit of (\ref{24b}) is immediate (this limit formula was presented
in \cite{Fo94x} but with (\ref{1.20x}) only implicit).

\begin{proposition}
We have
\begin{align} \label{4.10}
& K^{(SW_e)}_{\rm edge}(X,Y) := \lim_{N \to \infty} K^{(SW_e)}(x,y) \Big |_{x = - \pi N / L c + X \atop y = - \pi N / L c + Y} 
   \nonumber \\
  & \quad  = \Big ( {c \over \pi} \Big )^{1/2} {1 \over (q;q)_\infty} {e^{-c (X^2 + Y^2)/2} \over  \sinh (\pi (X - Y)/L) } 
 \nonumber  \\
& \qquad  \times \Big ( A_q(q^{1/2} e^{2 \pi X/L};q)    A_q(q^{-1/2} e^{2 \pi Y/L};q) -
 A_q(q^{1/2} e^{2 \pi Y/L};q)    A_q(q^{-1/2} e^{2 \pi X/L};q) \Big ).
\end{align}
\end{proposition}

\begin{remark}
 1.~The order of the remainder in (\ref{1.20x}) implies that the convergence to the limit happens at a rate
 proportional to $q^N$ and is thus exponentially fast. \\
 2.~It is noted in  \cite{Fo94x} that the first evaluation in Proposition \ref{p4.1} can be reclaimed
for $N$ even (odd)  from (\ref{4.10}) by making the replacements $X \mapsto M + x$ ($M - 1/2 + x$),
$Y \mapsto M + y$ ($M - 1/2 + y$) and taking the limit $M \to \infty$. \\
3.~The edge scaling limit of the Christoffel-Darboux kernel for both the $q$-Hermite 
and $q$-Laguerre orthogonal polynomial
systems has been shown in \cite{Is05} to have explicit forms also involving the function
$A_q(z)$, but which are distinct from each other, and distinct from (\ref{4.10}). \\
4.~According to (\ref{24b}) and (\ref{24}), a direct analysis of the right edge scaling limit requires
the large $N$ form of $S_N(q^{-2N} z; q)$ with $z, q$ fixed. This expansion can be found
in \cite{Is05}, which is obtained from the series (\ref{1.20}) by first replacing $\nu$ by $l - \nu$. As
commented in \cite{Wo18}, the latter replacement implies the symmetry
$$
S_n (z;q) = ( - z q^n)^n S_n \Big ( {1 \over z q^{2n}}; q \Big ).
$$
This symmetry used in (\ref{24b}) and (\ref{24}) maps the right edge to the left edge, showing
both are equivalent, as can be anticipated from (\ref{1.15}).
 \end{remark}
 
 Associated with the (left) edge scaled kernel $K^{(SW_e)}_{\rm edge}$ is the gap probability
 \begin{equation}\label{m1}  
 E^{(SW_e)}_{\rm edge}(0;(-\infty,s)) = \det \Big ( \mathbb I - \mathbb K^{(SW_e)}_{\rm edge} \Big |_{(-\infty,s)} \Big ).
 \end{equation}
 Here $\mathbb  K^{(SW_e)}_{\rm edge}  \Big |_{(-\infty,s)} $ denotes the integral operator on $(-\infty,s)$ with kernel
 $K^{(SW_e)}_{\rm edge} (X,Y)$.
 The gap probability in turn  determines the probability density function of the scaled position
 of the leftmost particle, $p_{\rm left}^{(SW_e)}(s)$ say,
 by a simple differentiation
  \begin{equation}\label{m2} 
 p_{\rm left}^{(SW_e)}(s)  = -   {d \over ds}    E^{(SW_e)}_{\rm edge}(0;(-\infty,s)).
  \end{equation}
  
  The gap probability admits an expansion in terms of the edge correlations, obtained by substituting $K^{(SW_e)}_{\rm edge} $ in
  (\ref{23}) (see e.g.~\cite[Eq.~(9.4)]{Fo10})
  \begin{multline}\label{m3}
   E^{(SW_e)}_{\rm edge}(0;(-\infty,s))  \\
   = 1 - \int_{-\infty}^s  K^{(SW_e)}_{\rm edge} (x,x)  \, dx +  {1 \over 2}  \int_{-\infty}^s dx_1   \int_{-\infty}^s dx_2 \,
   \det \begin{bmatrix} K^{(SW_e)}_{\rm edge} (x_1,x_1) &  K^{(SW_e)}_{\rm edge} (x_1,x_2)  \\
  K^{(SW_e)}_{\rm edge} (x_2,x_1)  &  K^{(SW_e)}_{\rm edge} (x_2,x_2)  \end{bmatrix} - \cdots
 \end{multline}
 It follows from (\ref{m2}) and (\ref{m3}) that
   \begin{equation}\label{m4} 
 p_{\rm left}^{(SW_e)}(s)  \mathop{\sim}\limits_{s \to - \infty}     K^{(SW_e)}_{\rm edge} (s,s) ,
\end{equation}
and so recalling (\ref{4.10}) and (\ref{1.20y}) we have that $ p_{\rm left}^{(SW_e)}(s)$ exhibits a
leading order Gaussian decay $e^{-c s^2}$ in its left tail. This is in agreement with the
prediction from \cite[Eq.~(8)]{DKMSS17} for the classical one-dimensional Coulomb gas
in a confining harmonic potential of strength $1/4$ (therefore, upon comparing with
(\ref{2.1}), we must set $c = 1/2$ for the results of \cite{DKMSS17}). This same reference
also predicts the leading behaviour in the right tail 
 \begin{equation}\label{m5} 
 p_{\rm left}^{(SW_e)}(s)  \mathop{\sim}\limits_{s \to  \infty}  \exp \Big ( - s^3/24 \alpha + O(s^2) \Big ), \qquad \alpha =  \pi / L.
\end{equation}
While we know of no direct way to establish this result from (\ref{m1}), by modifying the statistical
mechanics model   (\ref{2.1}) to its natural two-dimensional extension 
(recall  the discussion of \S \ref{S2.1}), an analytic derivation of (\ref{m5}) is possible; see the Appendix.
 
 \subsection{Scaling limit of $A_q(z)$ and the Airy kernel}\label{S4.3}
 It is well known that for the PDF (\ref{1.1}) the Christoffel--Darboux kernel $K^{(G)}(x,y)$ has a bulk scaling
 limit equal to the RHS of (\ref{25}); see e.g.~\cite[Ch.~7]{Fo10}. We know from (\ref{25}) that an
 appropriate $L \to \infty$ scaling of $K_{\rm bulk}^{(SW_e)}$ reclaims this functional form.
 To leading order the edges of the spectrum for (\ref{1.1}) interpreted as an
 eigenvalue PDF are
 at $\pm \sqrt{2N}$, and $K^{(G)}$ admits the edge scaling limit (\ref{SA}).
 The question to be addressed is to identify a scaling of
 $X,Y$ such that for  $L \to \infty$ the kernel $ K^{(SW_e)}_{\rm edge}(X,Y)$ reduces to
 the Airy kernel (\ref{SA}). This is answered in Proposition \ref{P1d}.
 For the proof, appropriate asymptotic properties of the special function $A_q(z)$
(\ref{1.20y}) are required. 

Before introducing these asymptotic properties, which fortunately are available in the
literature \cite{LW13,HP15,Ha17}, some contextual information relating to $A_q(z)$ is
appropriate. Firstly, names associated with $A_q(z)$ are the Ramanujan function, and the
$q$-Airy function; see \cite{Is05} for the underlying reasons. In relation to the latter, it
is important to be aware that there are other candidates which qualify for the title
of $q$-Airy functions, see in particular \cite{HKW06} for special function solutions of the
$q$-Painlev\'e II equation which are shown to limit to the Airy function. Actually these
various candidates can be related \cite{Mo14}; see also \cite{IZ07}. We remark too that
$A_q(z)$ satisfies the functional equation ($q$-difference equation)
$$
q x u (q^2 x) - u(qx) + u(x) = 0,
$$
and can also be regarded as a degeneration of the basic hypergeometric function (\ref{14a+}),
being given by
$$
A_q(x) =   {}_0 \phi_1 \Big ( {\underline{\: \:} \atop 0} \Big | q; - {qx } \Big ), \qquad
  {}_0 \phi_1 \Big ( {  \underline{\: \:} \atop b} \Big | q;z \Big )   := \sum_{n=0}^\infty {q^{n(n-1)} \over (q;q)_n (b;q)_n} z^n.
  $$
  
\subsubsection{Proof of Proposition \ref{P1d}} 
The key to establishing (\ref{16x}) is an asymptotic formula contained in \cite{Ha17}.
Specifically, with
\begin{equation}\label{4.20}  
q = e^{- \epsilon}, \qquad \alpha = 1 - 4z, \qquad \beta = {\log (z)^2 \over 4} + {\pi^2 \over 12},
\end{equation}
we have from \cite[Theorem 4.7.1, after simplification of (4.55)]{Ha17} (see also \cite{HP15}, and
compare the leading term with \cite[Th.~2]{LW13}), that for $\epsilon \to 0^+$, with $\alpha/ \epsilon^{2/3}$
fixed
\begin{equation}\label{4.21}  
A_q(z) = {1 \over 2} (q;q)_\infty e^{\beta/\epsilon} \Big (
{\rm Ai} \Big ( {\alpha \over \epsilon^{2/3}} \Big )  \epsilon^{1/3}  - {\rm Ai}' \Big ( {\alpha \over \epsilon^{2/3}} \Big )  \epsilon^{2/3} \Big ) \Big (
1 + O( \epsilon) \Big ).
\end{equation}

To see the relevance of (\ref{4.21}) in relation to (\ref{4.10}), note that with $X(x,l)$ as in (\ref{16r})
$$
e^{(2 \pi / L) X(x,L)} = {1 \over 4} (1 - \epsilon^{2/3} x) + O(\epsilon^{4/3}),
$$
and so considering $A_q( e^{(2 \pi / L) X(x,L)})$ leads to (\ref{4.21}) with $\alpha=\epsilon^{2/3} x$.
A (minor) detail is that (\ref{4.10}) requires not $A_q( e^{(2 \pi / L) X(x,L)})$ but rather
$A_q(q^{\pm 1/2} e^{(2 \pi / L) X(x,L)})$. By a first order Taylor expansion, this changes the prefactor
of ${\rm Ai}'$ in (\ref{4.21}) from $-1$ to $- {1 \over 2}$ and $- {3 \over 2}$ respectively. Noting this,
substituting in (\ref{4.10}), and recalling too the standard $\epsilon \to 0^+$ asymptotic formula
$$
\log (q;q)_\infty = - {\pi^2 \over 6 \epsilon} + \log \sqrt{2 \pi \over \epsilon} + O(\epsilon),
$$
which follows from the functional equation for the Dedekind eta function \cite{WW65}, we find after some minor simplification that (\ref{16x}) results.

     \subsection*{Acknowledgements}
	This research is part of the program of study supported
	by the Australian Research Council Centre of Excellence ACEMS.

	\appendix
\section*{Appendix}\label{A1}
\renewcommand{\thesection}{A} 
\setcounter{equation}{0}
The statistical mechanical model on a cylinder, as referred to in Section \ref{S2.1}, was shown to be
exactly solvable at $\beta = 2$ by Choquard \cite{Ch81}, and further elaborated on in \cite{CFS83}.
Define
\begin{equation}\label{WL}
W_L := {1 \over 4} \sum_{j=1}^N x_j^2 + \sum_{1 \le j < k \le N} \Phi((x_j,y_j),(x_j',y_j')),
\end{equation}
with $\Phi$ given by (\ref{2.2}), so that the Boltzmann factor with $\beta = 2$ is $e^{-2 W_L}$.
For this system, the probability $E^{(2d)}(0,(-\infty,s) \times (0,L))$ say that there are no particles in the
region $x \in (-\infty, s)$, $y \in [0,L)$, is given by
\begin{equation}\label{WL1}
E^{(2d)}(0,(-\infty,s) \times (0,L)) = Q_N(s)/ Q_N(-\infty),
\end{equation}
where
\begin{equation}\label{WL2}
Q_N(s) := \int_s^{\infty} dx_1 \cdots  \int_s^{\infty} dx_N \int_0^L dy_1 \cdots \int_0^L dy_N \, e^{-2 W_L}.
\end{equation}

The integration technique of  \cite{Ch81}, \cite{CFS83}, which uses the fact that with $z_l := e^{2 \pi i (y_l + i x_l) / L}$,
$$
 e^{-2 W_L} \propto \prod_{l=1}^N e^{- \sum_{j=1}^N (x_j^2 - 4\pi  (N - 1)  x_j / L)/2}
 \prod_{1 \le j < k \le N} (z_k - z_j)( \bar{z}_k - \bar{z}_j),
 $$
 then replaces each product with a Vandermonde determinant according to (\ref{1.8}), allows $Q_N(s)$ to be
 computed explicitly as a product of one-dimensional integrals. Substituting in (\ref{WL1}) shows
 \begin{equation}\label{WL3}
 E^{(2d)}(0,(-\infty,s) \times (0,L)) = \prod_{l=1}^N \sqrt{2 \over \pi} \int_s^\infty e^{- (x - 2 \pi (N - 2l + 1)/L)^2/2} \, dx.
 \end{equation}
 In the notation of (\ref{2.1}), (\ref{WL}) corresponds to a harmonic potential of strength $c=1/2$, so
 from the discussion of \S \ref{S4.2} the left edge occurs at $s^*:= -2 \pi N / L$. It follows from (\ref{WL3})
 that
 \begin{equation}\label{WL4}
  E^{(2d)}_{\rm edge}(0,(s,\infty)) := \lim_{N \to \infty}  E^{(2d)}(0,(s^*+s,\infty) \times (0,L)) = \prod_{l=0}^\infty  \sqrt{2 \over \pi} \int_s^\infty e^{- (x - 2 \pi ( 2l + 1)/L)^2/2} \, dx.
  \end{equation}
  For $s$ large in this expression, see that the term $l$ in the product contributes of order $e^{- (s - 2 \pi ( 2l + 1)/L)^2/2}$ for $l$ up to the value
  $s L / 2 \pi$ (appropriately rounded), and unity after this. Hence to leading order
   \begin{equation}\label{WL5}
  E^{(2d)}_{\rm edge}(0,(s,\infty)) \mathop{\sim}\limits_{s \to \infty} \prod_{l=0}^{[s L / 2 \pi]} e^{- (s - 2 \pi ( 2l + 1)/L)^2/2} \sim e^{- s^3 L / 24 \pi},
    \end{equation}
    where the final asymptotic expression follows by summing the exponents in the expression before, observing it can be
    written as a Riemann sum. This gives agreement with (\ref{m5}).

\end{document}